\documentclass[11 pt, onecolumn, journal]{IEEEtran}

\usepackage{amsfonts,amssymb,amsmath,epsfig,graphicx}
\usepackage{cite}
\usepackage{url}
\usepackage{color}
\usepackage{mystyle}

\newcommand{\perr}{P_{\rm err}(R,\snr)}
\newcommand{\pout}{P_{\rm out}(R,\snr)}
\newcommand{\pmmse}{P^{\rm mmse}_{\rm out}(R,\snr)}
\newcommand{\pzf}{P^{\rm zf}_{\rm out}(R,\snr)}
\newcommand{\mmsesnr}{\gamma_{\rm mmse}(\snr,\bh)}
\newcommand{\zfsnr}{\gamma_{\rm zf}(\snr,\bh)}

\begin{document}

\allowdisplaybreaks
\title{Diversity Analysis of Symbol-by-Symbol Linear Equalizers}

\author{Ali Tajer\authorrefmark{1}\footnote{\authorrefmark{1}
Electrical Engineering Department, Princeton University, Princeton, NJ 08544.}\qquad Aria Nosratinia \authorrefmark{2}\footnote{\authorrefmark{2} Electrical Engineering Department, University of Texas at Dallas, Richardson, TX 75083.} \qquad Naofal Al-Dhahir\authorrefmark{2}}
\maketitle

\begin{abstract}

In frequency-selective channels linear receivers enjoy significantly-reduced complexity compared with maximum likelihood receivers at the cost of performance degradation which can be in the form of a loss of the inherent frequency diversity order or reduced coding gain. This paper demonstrates that the minimum mean-square error symbol-by-symbol linear equalizer incurs no diversity loss compared to the maximum likelihood receivers. In particular, for a channel with memory $\nu$, it achieves the full diversity order of ($\nu+1$) while the zero-forcing symbol-by-symbol linear equalizer always achieves a diversity order of one.
\end{abstract}


\section{Introduction}
\label{sec:intro}

In broadband wireless communication systems, the coherence bandwidth of the fading channel is significantly less than the transmission bandwidth. This results in inter-symbol interference (ISI) and at the same time provides frequency diversity that can be exploited at the receiver to enhance transmission reliability \cite{Proakis:book}. It is well-known that for Rayleigh {\em flat}-fading channels, the
error rate decays only linearly with signal-to-noise ratio ($\snr$)~\cite{Proakis:book}. For frequency-selective channels, however, proper exploitation of the available frequency diversity forces the error probability to decay at a possibly higher rate and, therefore, can potentially achieve higher diversity gains, depending on the detection scheme employed at the receiver.

While maximum likelihood sequence detection (MLSD)~\cite{forney:ML} achieves optimum performance over ISI channels, its complexity (as measured by the number of MLSD trellis states) grows \emph{exponentially} with the spectral efficiency and the channel memory. As a low-complexity alternative, filtering-based symbol-by-symbol equalizers (both linear and decision feedback) have
been widely used over the past four decades (see \cite{qureshi:adaptive} and \cite{vitetta} for excellent tutorials). Despite their long history and successful commercial deployment, the performance of symbol-by-symbol linear equalizers over wireless fading channels is not fully characterized. More specifically, it is not known whether their observed sub-optimum performance is due to their inability to fully exploit the channel's frequency diversity or due to a degraded performance in combating the residual inter-symbol interference. Therefore, it is of paramount importance to investigate the frequency diversity order achieved by linear equalizers, which is the subject of this paper. Our analysis shows that while single-carrier infinite-length symbol-by-symbol minimum mean-square error (MMSE) linear equalization achieves full frequency diversity, zero-forcing (ZF) linear equalizers cannot exploit the frequency diversity provided by frequency-selective channels.

A preliminary version of the results of this paper on the MMSE linear equalization has partially appeared in  [5] and the proofs available in \cite{ali:ISIT07_1} are skipped and referred to wherever necessary. The current paper provides two key contributions beyond [5]. First, the diversity analysis of ZF equalizers is added. Second, the MMSE analysis in [5] lacked a critical step that was not rigorously complete; the missing parts that have key role in analyzing the diversity order are provided in this paper.

\section{System Descriptions}
\label{sec:descriptions}
\subsection{Transmission Model}
\label{sec:transmission}

Consider a quasi-static ISI wireless fading channel with memory length $\nu$ and channel impulse response (CIR) denoted by $\bh=[h_0,\dots,h_{\nu}]$. Without loss of generality, we restrict our
analyses to CIR realizations with $h_0\neq 0$. The output of the channel at time $k$ is given by
\begin{equation}\label{eq:model_time}
    y_k=\sum_{i=0}^{\nu}h_ix_{k-i}+n_k\ ,
\end{equation}
where $x_k$ is the input to the channel at time $k$ satisfying the
power constraint $\mathbb{E}[|x_k|^2]\leq P_0$ and $n_k$ is
the additive white Gaussian noise term distributed as
$\mathcal{N}_\mathbb{C}(0,N_0)$\footnote{$\mathcal{N}_\mathbb{C}(a,b)$
denotes a complex Gaussian distribution with mean $a$ and variance
$b$.}. The CIR coefficients $\{h_i\}_{i=0}^\nu$ are distributed
independently with $h_i$ being distributed as
$\mathcal{N}_\mathbb{C}(0,\lambda_i)$. Defining the $D$-transform of the input sequence $\{x_k\}$ as $X(D)=\sum_k x_kD^k$, and similarly defining $Y(D), H(D)$, and $Z(D)$, the baseband input-output model can be cast in the $D$-domain as $Y(D)=H(D)\cdot X(D)+Z(D)$. The superscript $*$ denotes complex conjugate and we use the shorthand $D^{-*}$ for $(D^{-1})^*$. We define $\snr\dff\frac{P_0}{N_0}$ and say that the functions $f(\snr)$ and $g(\snr)$ are \emph{exponentially equal}, indicated by $f(\snr)\doteq g(\snr)$, when
\begin{equation}
    \label{eq:exp} \lim_{\snr\rightarrow\infty}\frac{\log f(\snr)}{\log
    \snr}=\lim_{\snr\rightarrow\infty}\frac{\log g(\snr)}{\log
    \snr}\ .
\end{equation}
The operators $\dotlt$ and $\dotgt$ are defined in a similar fashion. Furthermore, we say that the \emph{exponential order} of $f(\snr)$ is $d$ if $f(\snr)\doteq \snr^d$.

\subsection{Linear Equalization}
\label{sec:equalization}

The zero-forcing (ZF) linear equalizers are designed to produce an
ISI-free sequence of symbols and ignore the resulting noise enhancement. By taking into account the {\em combined} effects of the ISI channel and its corresponding matched-filter, the ZF linear equalizer in the $D$-domain is given by~\cite[Equation~(3.87)]{cioffi}
\begin{equation}\label{eq:zf_eq}
      W_{\rm zf}(D)= \frac{\|\bh\|}{H(D)H^*(D^{-*})}\ ,
\end{equation}
where the $\|\bh\|$ is the $\ell_2$-norm of $\bh$, i.e.,
$\|\bh\|^2=\sum_{i=0}^{\nu}|h_i|^2$. The variance of the noise seen at
the output of the ZF equalizer is the key factor in the performance of the
equalizer and is given by
\begin{equation}\label{eq:zf_var}
    \sigma^2_{\rm
    zf}\dff\frac{1}{2\pi}\int_{-\pi}^{\pi}\frac{N_0}{|H(e^{-ju})|^2}\;du\ .
\end{equation}
Therefore, the decision-point signal-to-noise ratio for any CIR realization $\bh$ and $\snr=\frac{P_0}{N_0}$ is
\begin{equation}\label{eq:zf_snr}
    \zfsnr\dff\snr
    \bigg[\frac{1}{2\pi}\int_{-\pi}^{\pi}\frac{1}{|H(e^{-ju})|^2}\;du\bigg]^{-1}.
\end{equation}
MMSE linear equalizers are designed to strike a balance between ISI reduction and noise enhancement through minimizing the combined residual ISI and noise level. Given the combined effect of the ISI channel and its corresponding matched-filter, the MMSE linear equalizer in the $D$-domain is~\cite[Equation~(3.148)]{cioffi}
\begin{eqnarray} \label{eq:mmse_eq}
  W_{\rm mmse}(D)= \frac{\|\bh\|}{H(D)H^*(D^{-*})+\snr^{-1}}\ .
\end{eqnarray}
The variance of the residual ISI and the noise variance as seen at the output of the equalizer is
\begin{equation}\label{eq:mmse_var}
    \sigma^2_{\rm
    mmse}\dff\frac{1}{2\pi}\int_{-\pi}^{\pi}\frac{N_0}{|H(e^{-ju})|^2+\snr^{-1}}\;du\ .
\end{equation}
Hence, the \emph{unbiased}\footnote{All MMSE equalizers are biased. Removing the bias decreases the decision-point signal-to-noise ratio by $1$ (in linear scale) but improves the error probability \cite{CDEF}. All the results provided in this paper are valid for biased receivers as well.} decision-point signal-to-noise ratio at for any CIR realization $\bh$ and $\snr$ is
\begin{eqnarray}
  \label{eq:mmse_snr}\mmsesnr
  \dff\bigg[\frac{1}{2\pi}\int_{-\pi}^{\pi}\frac{1}{\snr|H(e^{-ju})|^2+1}\;
  du\bigg]^{-1}-1\ .
\end{eqnarray}

\subsection{Diversity Gain}
\label{sec:diversity}

For a transmitter sending information bits at spectral efficiency $R$ bits/sec/Hz, the system is said to be in \emph{outage} if the ISI channel is faded such that it cannot sustain an arbitrarily reliable communication at the  intended communication spectral efficiency $R$, or equivalently, the mutual information  $I(x_k,\tilde y_k)$ falls below the target spectral efficiency $R$, where $\tilde y_k$ denotes the equalizer output. The probability of such outage for the signal-to-noise ratio $\gamma(\snr,\bh)$ is
\begin{equation}
    \label{eq:out}\pout\dff P_{\bh}\bigg(\log\Big[1+\gamma(\snr,\bh)\Big]<R\bigg)\ ,
\end{equation}
where the probability is taken over the ensemble of all CIR realizations $\bh$. The outage probability at high transmission powers ($\snr\rightarrow\infty$) is closely related to the \emph{average pairwise error probability}, denoted by $\perr$, which is the probability that a transmitted codeword $\bc_i$ is erroneously detected in favor of another codeword $\bc_j$, $j\neq i$, i.e.,
\begin{equation}\label{eq:perr}
    \perr\dff \bbe_{\bh}\bigg[ P\Big(\bc_i\rightarrow\bc_j\med {\bh}\Big)\bigg]\ .
\end{equation}
When deploying channel coding with arbitrarily long code-length, the outage and error probabilities decay at the same rate with increasing $\snr$ and have the same exponential order \cite{zheng:IT03} and therefore
\begin{equation}\label{eq:equality}
    \pout\doteq\perr\ .
\end{equation}
This is intuitively justified by noting that in high $\snr$ regimes, the effect of channel noise is diminishing and the dominant source of erroneous detection is channel fading which, as mentioned above, is also the source of outage events. As a result, in our setup, diversity order which is the negative of the exponential order of the average pairwise error probability $\perr$ is computed as
\begin{equation}\label{eq:diversity}
    d=-\lim_{\snr\rightarrow\infty}\frac{\log \pout}{\log\snr}\ .
\end{equation}

\section{Diversity Order of MMSE Linear Equalization}
\label{sec:mmse}
The main result of this paper for the MMSE linear equalizers is given in the following theorem.
\begin{theorem}\label{th:mmse}
For an ISI channel with channel memory length $\nu\geq 1$, and symbol-by-symbol MMSE linear equalization we have
\begin{equation*}
    P_{\rm err}^{\rm mmse}(R,\snr)\doteq\snr^{-(\nu+1)}.
\end{equation*}
\end{theorem}

The sketch of the proof is as follows. First, we find a lower bound on the unbiased decision-point signal-to-noise ratio ($\snr$) and use this lower bound to show that for small enough spectral efficiencies a full diversity order of $(\nu+1)$ is achievable. The proof of the diversity gain for low spectral efficiencies is offered in Section~\ref{sec:low}. In the second step, we show that increasing the spectral efficiency to any arbitrary level does not incur a diversity loss, concluding that MMSE linear equalization is capable of collecting the full frequency diversity order of ISI channels. Such generalization of the results presented in Section~\ref{sec:low} to arbitrary spectral efficiencies is analyzed in Section~\ref{sec:full}.

\subsection{Full Diversity for Low Spectral Efficiencies}
\label{sec:low}
We start by showing that for arbitrarily small data transmission spectral efficiencies, $R$, full diversity is achievable. Corresponding to each CIR realization $\bh$, we define the function  $f(\bh,u)\dff|H(e^{-ju})|^2-\|\bh\|^2$ for which after some simple manipulations we have
\begin{align}\label{eq:f}
   f(\bh,u)&= \sum_{k=-\nu}^\nu c_k\; e^{jku}\ ,\;\;\mbox{where}\;\; c_0=0,\;\; c_{-k}=c^*_k, \;\;  c_k=\sum_{m=0}^{\nu-k}h_mh^*_{m+k}\; \;\; \mbox{for} \;\;\;\; k\in\{1,\dots, \nu\}\ .
\end{align}
Therefore, $f(\bh,u)$ is a trigonometric polynomial of degree $\nu$ that is periodic with period $2\pi$ and in the open interval $[-\pi,\pi]$ has at most $2\nu$ roots~\cite{powell:book}. Corresponding to the CIR realization $\bh$ we define the set
\begin{equation*}
    {\cal D}(\bh)\dff\{u\in[-\pi,\pi]\;:\; f(\bh,u)>0\}\ ,
\end{equation*}
and use the convention $|{\cal D}(\bh)|$ to denote the measure of ${\cal D}(\bh)$, i.e., the aggregate lengths of the intervals over which $f(\bh,u)$ is strictly positive. In the following lemma, we obtain a lower bound on $|{\cal D}(\bh)|$ which is instrumental in finding a lower bound on $\mmsesnr$.

\begin{lemma}
\label{lemma:interval} There exists a real number $C>0$ such that for all non-zero CIR realizations $\bh$, i.e. $\forall\bh\neq\boldsymbol 0$, we have that $|{\cal D}(\bh)|\geq C \left(2(2\nu+1)^3\right)^{-\frac{1}{2}}$.
\end{lemma}
\begin{proof}
According to \eqref{eq:f} we immediately have $\int_{-\pi}^{\pi}f(\bh,u)\;du=0$. By invoking the definition of ${\cal D}(\bh)$ and noting that $[-\pi,\pi]\backslash{\cal D}(\bh)$ includes the values of $u$ for which $f(\bh,u)$ is negative, we have
\begin{equation}\label{eq:f_int}
    \int_{{\cal D}(\bh)}f(\bh,u)\; du=-\int_{[-\pi,\pi]\backslash{\cal D}(\bh)}f(\bh,u)\; du\quad\Rightarrow\quad \int_{-\pi}^{\pi}|f(\bh,u)|\;du=2\int_{{\cal D}(\bh)}f(\bh,u)\; du\ .
\end{equation}
Also by noting that $f(\bh,u)=|H(e^{-ju})|^2-\|\bh\|^2$, $f(\bh,u)$ has clearly a real value for any $u$. Moreover, by invoking \eqref{eq:f} from the Cauchy-Schwartz inequality we obtain
\begin{equation}\label{eq:f_CS}
    f(\bh,u)\leq |f(\bh,u)|\leq\bigg(\sum_{k=-\nu}^\nu |c_k|^2\bigg)^{\frac{1}{2}}\bigg(\sum_{k=-\nu}^\nu |e^{jku}|^2\bigg)^{\frac{1}{2}} = \bigg(2(2\nu+1)\sum_{k=1}^\nu |c_k|^2\bigg)^{\frac{1}{2}}\ .
\end{equation}
Equations \eqref{eq:f_int} and \eqref{eq:f_CS} together establish that
\begin{equation}\label{eq:f_int_bound}
    |{\cal D}(\bh)|\geq \frac{1}{2}\bigg(2(2\nu+1)\sum_{k=1}^\nu |c_k|^2\bigg)^{-\frac{1}{2}}\; \int_{-\pi}^{\pi}|f(\bh,u)|\;du\ .
\end{equation}
Next we strive to find a lower bound on $\int_{-\pi}^{\pi}|f(\bh,u)|\;du$, which according to \eqref{eq:f} is equivalent to finding a lower bound on the $\ell_1$ norm of a sum of exponential terms. Obtaining lower bounds on the $\ell_1$ norm of exponential sums has a rich literature in the mathematical analysis and we use a relevant result in this literature that is related to Hardy's inequality \cite[Theorem 2]{mcgehse}.
\begin{theorem}
\emph{\cite[Theorem 2]{mcgehse}}
There is a real number $C>0$ such that for any given sequence of increasing integers $\{n_k\}$, and complex numbers $\{d_k\}$, and for any $N\in\mathbb{N}$ we have
\begin{equation}\label{eq:Hardy}
    \int_{-\pi}^{\pi}\bigg|\sum_{k=1}^Nd_k\;e^{jn_ku}\bigg|\;du\geq C\sum_{k=1}^N\frac{|d_k|}{k}\ .
\end{equation}
\end{theorem}
By setting $N=2\nu+1$ and $d_k=c_{k-(\nu+1)}$ and $n_k=k-(\nu+1)$ for $k\in\{1,\dots,2\nu+1\}$ from \eqref{eq:Hardy} it is concluded that there exists $C>0$ that for each set of $\{c_{-\nu},\dots, c_\nu\}$ we have
\begin{align}\label{eq:Hardy2}
     \int_{-\pi}^{\pi}|f(\bh,u)|\;du  \geq C\sum_{k=1}^{2\nu+1}\frac{|c_{k-(\nu+1)}|}{k}\geq  \frac{C}{2\nu+1}\sum_{k=1}^{2\nu+1}|c_{k-(\nu+1)}|= \frac{2C}{2\nu+1}\sum_{k=1}^{\nu}|c_{k}|\ ,
\end{align}
where the last equality holds by noting that $c_{-k}=c^*_k$ and $c_0=0$. Combining \eqref{eq:f_int_bound} and \eqref{eq:Hardy2}  provides
\begin{equation}\label{eq:f_int_bound2}
    |{\cal D}(\bh)|\geq C\left(2(2\nu+1)^3\right)^{-\frac{1}{2}} \underset{\geq 1}{\underbrace{\frac{\sum_{k=1}^\nu |c_k|}{\sqrt{\sum_{k=1}^\nu |c_k|^2}}}}\geq C\left(2(2\nu+1)^3\right)^{-\frac{1}{2}}\ ,
\end{equation}
which concludes the proof.
\end{proof}
Now by using Lemma \ref{lemma:interval} for any CIR realization $\bh$ and $\snr$ we find a lower bound on $\mmsesnr$  that depends on $\bh$ through $\|\bh\|$ only. By defining ${\cal D}^c(\bh)=[-\pi,\pi]\backslash{\cal D}(\bh)$ we have
\begin{align}
    \nonumber
    1+\mmsesnr&\overset{\eqref{eq:mmse_snr}}{=} \bigg[\frac{1}{2\pi}\int_{-\pi}^{\pi}\frac{1}{\snr|H(e^{ju})|^2+1}\;du\bigg]^{-1} = \bigg[\frac{1}{2\pi}\int_{-\pi}^{\pi}\frac{1}{\snr(f(\bh,u)+\|\bh\|^2)+1}\;du\bigg]^{-1} \\
    \nonumber &= \bigg[\frac{1}{2\pi}\int_{{\cal D}(\bh)}\frac{1}{\snr(\underset{>
    0}{\underbrace{f(\bh,u)}}+\|\bh\|^2)+1}\;du+ \frac{1}{2\pi}
    \int_{{\cal D}^c(\bh)}\frac{1}{\snr\underset{\geq
    0}{\underbrace{|H(e^{ju})|^2}}+1}\;du\bigg]^{-1}\\
    \nonumber &\geq \bigg[\frac{1}{2\pi}\int_{{\cal D}(\bh)}\frac{1}{\snr\|\bh\|^2+1}\;du+ \frac{1}{2\pi}
    \int_{{\cal  D}^c(\bh)}1\;du\bigg]^{-1}\\
    \nonumber &=
    \bigg[\frac{|{\cal D}(\bh)|}{2\pi}\cdot\frac{1}{
    \snr\|\bh\|^2+1} + \bigg(1-\frac{|{\cal D}(\bh)|}{2\pi}\bigg)\bigg]^{-1}\\
    \nonumber & = \bigg[1-\frac{|{\cal D}(\bh)|}{2\pi}\bigg(1-\frac{1}{\snr\|\bh\|^2+1}\bigg)\bigg]^{-1}\\
    \label{eq:mmse_snr_lb} & \overset{\eqref{eq:f_int_bound2}}{\geq}  \bigg[1-\frac{C\left(2(2\nu+1)^3\right)^{-\frac{1}{2}}}{2\pi}\bigg(1-\frac{1}{\snr\|\bh\|^2+1}\bigg)\bigg]^{-1} \ .
\end{align}
By defining $C'\dff \frac{C\left(2(2\nu+1)^3\right)^{-\frac{1}{2}}}{2\pi}$, for the outage probability corresponding to the target spectral efficiency $R$ we have
\begin{align}
    \nonumber \pmmse&\overset{\eqref{eq:out}}{=} P_{\bh}
    \bigg(1+\mmsesnr<2^R\bigg)\overset{\eqref{eq:mmse_snr_lb}}{\leq}
    P_{\bh}\bigg\{1-C'\bigg(1-\frac{1}{\snr\|\bh\|^2+1}\bigg)>2^{-R}\bigg\}\\
    \label{eq:mmse_out_lb1} &= P_{\bh}\bigg\{1-\frac{1-2^{-R}}{C'}<\frac{1}{\snr\|\bh\|^2+1}\bigg\}\ .
\end{align}
If
\begin{equation}\label{eq:rate1}
    1-\frac{1-2^{-R}}{C'}> 0\qquad\mbox{or equivalently}\qquad R<R_{\max}\dff\log_2\left(\frac{1}{1-C'}\right)\ ,
\end{equation}
then the probability term in \eqref{eq:mmse_out_lb1} can be restated as
\begin{equation}\label{eq:mmse_outlb2}
    P_{\bh}\bigg\{\snr\|\bh\|^2 <
    \frac{1-2^{-R}}
    {C'-(1-2^{-R})}\bigg\} = P_{\bh}\bigg\{\snr\|\bh\|^2 <
    \frac{2^{R}-1}
    {1-2^{R-R_{\max}}}\bigg\}\ .
\end{equation}
Therefore, based on (\ref{eq:mmse_out_lb1})-(\ref{eq:mmse_outlb2}) for all $0<R<R_{\max}$ we have
\begin{align}
    \nonumber \pmmse &\leq
    P_{\bh}\bigg\{\snr\|\bh\|^2<\frac{2^{R}-1}
    {1-2^{R-R_{\max}}}\bigg\} = P_{\bh}\bigg\{\snr\sum_{m=0}^{\nu}|h_m|^2<\frac{2^{R}-1}
    {1-2^{R-R_{\max}}}\bigg\}\\
    \label{eq:mmse_out_lb3} &\leq \prod_{m=0}^{\nu}P_{\bh}\bigg\{|h_m|^2<\frac{2^{R}-1}
    {\snr(1-2^{R-R_{\max}})}\bigg\}\doteq \snr^{-(\nu+1)}\ .
\end{align}
Therefore, for the spectral efficiencies $R\in(0, R_{\max})$ we have $\pmmse\;\dotlt\;\snr^{-(\nu+1)}$, which in conjunction with (\ref{eq:equality}) proves that $P_{\rm err}^{\rm mmse}\;\dotlt\;\snr^{-(\nu+1)}$, indicating that a diversity order of at least $(\nu+1)$ is achievable. On the other hand, since the diversity order cannot exceed the number of the CIR taps, the achievable diversity order is exactly $(\nu+1)$. Also note that the real number $C>0$ given in \eqref{eq:f_int_bound2} is a constant independent of the CIR realization $\bh$ and, therefore, $C'$ and, consequently, $R_{\max}$ are also independent of the CIR realization.
This establishes the proof of Theorem \ref{th:mmse} for the range of the spectral efficiencies $R\in(0, R_{\max})$, where $R_{\max}$ is fixed and defined in~(\ref{eq:rate1}).

\subsection{Full Diversity for All Rates}
\label{sec:full}
We now extend the results previously found for $R<R_{\max}$ to all
spectral efficiencies.

\begin{lemma}
\label{lemma:linear} For asymptotically large values of $\snr$,  $\mmsesnr$ varies linearly with $\snr$, i.e.,
\begin{equation*}
    \lim_{\snr \rightarrow \infty} \frac{\partial\;\mmsesnr}{\partial\;\snr}=s(\bh),\;\;\; \mbox{where}\;\;\;
    s(\bh):\mathbb{R}^{\nu+1}\rightarrow\mathbb{R}\ .
\end{equation*}
\end{lemma}

\begin{proof}
See Appendix \ref{app:lemma:linear}.
\end{proof}

\begin{lemma}
\label{lemma:limit} For the continuous random variable $X$, variable $y\in\mathbb{R}$, constants $c_1, c_2\in\mathbb{R}$ and function $G(X,y)$ continuous in $y$, we have
\begin{equation*}
    \lim_{y\rightarrow y_0}P_X\Big(c_1 \leq G(X,y) \leq
    c_2\Big)=P_X\Big(c_1 \leq \lim_{y\rightarrow y_0}G(X,y) \leq
    c_2\Big)\ .
\end{equation*}
\end{lemma}
\begin{proof}
Follows from Lebesgue's Dominated Convergence theorem~\cite{bartle:B1} and the same line of argument as in \cite[Appendix C]{ali:ISIT07_1}
\end{proof}

Now, we show that if for some spectral efficiency $R^{\dag}$ the achievable diversity order is $d$, then for all spectral efficiencies {\em up} to $R^{\dag}+1$, the same diversity order is achievable. By induction, we conclude that the diversity order remains unchanged by changing the data spectral efficiency $R$. If for the spectral efficiency $R^{\dag}$, the negative of the exponential order of the outage probability is $d$, i.e.,
\begin{equation}
    \label{eq:induction}
    P_{\bh}\bigg(\log\Big[1+\mmsesnr\Big]<R^{\dag}\bigg)\doteq\snr^{-d},
\end{equation}
then by applying the results of Lemmas \ref{lemma:linear} and \ref{lemma:limit} for the target spectral efficiency $R^{\dag}+1$ we get
\begin{align}
    \nonumber \pmmse&=P_{\bh}\bigg(\log\Big[1+\mmsesnr\Big]<R^{\dag}+1\bigg)= P_{\bh}\left(1+\mmsesnr<2^{R^\dag+1}\right)\\
    \label{eq:induction_1} &\doteq P_{\bh}\left(\snr\;s(\bh)<2^{R^\dag+1}\right)= P_{\bh}\left(\Big(\frac{\snr}{2}\Big) s(\bh)<2^{R^\dag}\right)\\
   \label{eq:induction_2}  &\doteq P_{\bh}\left(1+\gamma_{\rm
    mmse}\Big(\frac{\snr}{2},\bh\Big)<2^{R^\dag}\right)\doteq P_{\bh}\bigg(\log\Big[1+\gamma_{\rm
    mmse}\Big(\frac{\snr}{2},\bh\Big)\Big]<R^{\dag}\bigg)\\
    \label{eq:induction_3}
    &\doteq\Big(\frac{\snr}{2}\Big)^{-d}за\doteq\snr^{-d}\ .
\end{align}
Equations~(\ref{eq:induction_1})~and~(\ref{eq:induction_2}) are
derived as the immediate results of
Lemmas~\ref{lemma:linear}~and~\ref{lemma:limit} that enable
interchanging the probability and the limit and also show that
$\mmsesnr\doteq \snr\cdot
s(\bh)$. Equations~(\ref{eq:induction})-(\ref{eq:induction_3}) imply
that the diversity orders achieved for the spectral efficiencies up to $R^\dag$ and the spectral efficiencies up to $R^\dag+1$ are the same. As a result, any arbitrary spectral efficiency exceeding $R_{\max}$ achieves the same spectral efficiency as the spectral efficiencies
$R\in(0,R_{\max})$ and, therefore, for any arbitrary spectral efficiency $R$, full diversity is achievable via MMSE linear equalization which completes the proof. Figure \ref{fig:1} depicts our simulation results for the pairwise error probabilities for two ISI channels with memory lengths $\nu=1$ and 2 and MMSE equalization. For each of these channels we consider signal transmission with spectral efficiencies $R=(1,2,3,4)$ bits/sec/Hz. The simulation results confirm that for a channel with two taps the achievable diversity order is two irrespective of the data spectral efficiency. Similarly, it is observed that for a three-tap channel the achievable diversity order is three.

\section{Diversity Order of ZF linear Equalization}
\label{sec:zf}
In this section, we show that the diversity order achieved by zero-forcing linear equalization, unlike that achievable with MMSE equalization, is independent of the channel memory length and is always 1.
\begin{lemma}
\label{lemma:zf} For any arbitrary set of normal complex Gaussian
random variables $\bmu\dff(\mu_1,\dots,\mu_m)$ (possibly correlated) and for any $B\in\mathbb{R}^+$ we have
\begin{equation}
    \label{eq:lemma:zf1}
    P_{\bmu}\bigg(\sum_{k=1}^m\frac{1}{\snr|\mu_k|^2}>B\bigg)\; \dotgt\; \snr^{-1}\ .
\end{equation}
\end{lemma}
\begin{proof}
Define $W_k \dff -\frac{\log|\mu_k|^2}{\log\snr}$. Since $|\mu_k|^2$ has exponential distribution, it can be shown that for any $k$ the cumulative density  function (CDF) at the asymptote of high values of $\snr$ satisfies \cite{azarian:IT05}
\begin{equation}\label{eq:W}
    1-F_{W_k}(w)\doteq\snr^{-w}\ .
\end{equation}
Thus, by substituting $|\mu_k|^2\doteq\snr^{1-W_k}$ based on (\ref{eq:W}) we find that
\begin{align}
    \label{eq:lemma:zf2} P_{\bmu}\bigg(\sum_{k=1}^m\frac{1}{\snr|\mu_k|^2}>B\bigg)&\doteq
    P_{\bW}\bigg(\sum_{k=1}^m\snr^{W_k-1}>B \bigg)
    \doteq P_{\bW}(\max_k W_k>1)\\
    \label{eq:lemma:zf3} &\geq P_{W_k}(W_k>1)=1-F_{W_k}(1)\doteq \snr^{-1}\ .
\end{align}
Equation~(\ref{eq:lemma:zf2}) holds as the term $\snr^{\max W_k-1}$ is
the dominant term in the summation $\sum_{k=1}^m\snr^{W_k-1}$. Also, the transition from (\ref{eq:lemma:zf2}) to (\ref{eq:lemma:zf3}) is justified by noting that $\max_kW_k\geq W_k$ and the last step is derived by taking into account \eqref{eq:W}.
\end{proof}


\begin{theorem}
\label{th:zf} The diversity order achieved by symbol-by-symbol ZF linear equalization is one, i.e.,
\begin{equation*}
    P_{\rm err}^{\rm zf}(R,\snr)\doteq \snr^{-1}
\end{equation*}
\end{theorem}
\begin{proof}
By recalling the decision-point signal-to-noise ratio of ZF equalization given in~(\ref{eq:zf_snr}) we have
\begin{align}
    \label{eq:zf1} \pzf&=P_{\bh}\Big(\zfsnr<2^R-1\Big)=P_{\bh}\bigg\{\Big[\frac{1}{2\pi}\int_{\pi}^{\pi}\frac{1}{\snr|H(e^{-ju})|^2}\;du\Big]^{-1}<2^R-1\bigg\}\\
    \label{eq:zf2} &=P_{\bh}\bigg\{\lim_{\Delta\rightarrow
    0}\bigg[\sum_{k=0}^{\lfloor 2\pi/\Delta\rfloor}\frac{\Delta}{\snr|H(e^{-j(-\pi+k\Delta)})|^2}\bigg]^{-1}<\frac{2^R-1}{2\pi}\bigg\}\\
    \label{eq:zf3} &=\lim_{\Delta\rightarrow
    0}P_{\bh}\bigg\{\bigg[\sum_{k=0}^{\lfloor 2\pi/\Delta\rfloor}\frac{\Delta}{\snr
    |H(e^{-j(-\pi+k\Delta)})|^2}\bigg]^{-1}<\frac{2^R-1}{2\pi}\bigg\}\\
     \label{eq:zf4}&=\lim_{\Delta\rightarrow
    0}P_{\bh}\bigg\{\sum_{k=0}^{\lfloor 2\pi/\Delta\rfloor}\frac{\Delta}{\snr|H(e^{-j(-\pi+k\Delta)})|^2}>\frac{2\pi}{2^R-1}\bigg\}  \; \dotgt \; \snr^{-1}\ .
\end{align}
Equation~(\ref{eq:zf2}) is derived by using Riemann integration, and (\ref{eq:zf3}) holds by using Lemma~\ref{lemma:limit} which allows for interchanging  the limit and the probability. Equation~(\ref{eq:zf4}) holds by applying Lemma~\ref{lemma:zf} on $\mu_k=H(e^{-j(-\pi+k\Delta)})$ which can be readily verified to have Gaussian distribution. Therefore, the achievable diversity order is 1.
\end{proof}

Figure \ref{fig:2} illustrates the pairwise error probability of two ISI channels with memory lengths $\nu=1$ and 2. The simulation results corroborate our analysis showing that the achievable diversity order is one, irrespective of the channel memory length or communication spectral efficiency.

\section{Conclusion}
\label{sec:conclusion}
We showed that infinite-length symbol-by-symbol MMSE linear equalization can fully capture the underlying frequency diversity of the ISI channel. Specifically, the diversity order achieved is equal to that of MLSD and in the high-$\snr$ regime, the performance of MMSE linear equalization and MLSD do not differ in diversity gain and the origin of their performance discrepancy is their ability to control the residual inter-symbol interference. We also show that the diversity order achieved by symbol-by-symbol ZF linear equalizers is always one, regardless of channel memory length.

\appendices

\section{Proof of Lemma \ref{lemma:linear}}
\label{app:lemma:linear}
We define $g(\bh,u)\dff|H(e^{-ju})|^2$ which has a finite number of zero by following the same line as for $f(\bh,u)$ in the proof of Lemma~{\ref{lemma:interval}}. By using (\ref{eq:mmse_snr}) we get
\begin{align}
    \nonumber\frac{\partial\;\mmsesnr}{\partial\;\snr}&= \frac{\partial}{\partial\;\snr}
    \left(\bigg[\frac{1}{2\pi}\int_{-\pi}^{\pi}\frac{1}{\snr g(\bh,u)+1}\;du\bigg]^{-1}-1\right)\\
    \label{eq:lemma:linear1} &=\bigg[\frac{1}{2\pi}\int_{-\pi}^{\pi}\frac{g(\bh,u)}{\left(\snr
    g(\bh,u)+1\right)^2}\;du\bigg]\cdot
    \bigg[\frac{1}{2\pi}\int_{-\pi}^{\pi}\frac{1}{\snr
    g(\bh,u)+1}\;du\bigg]^{-2}\\
    \label{eq:lemma:linear2} &=\bigg[\frac{1}{2\pi}\int_{g(\bh,u)\neq 0}\frac{g(\bh,u)}{\left(\snr
    g(\bh,u)+1\right)^2}\;du\bigg]\cdot
    \bigg[\frac{1}{2\pi}\int_{g(\bh,u)\neq 0}\frac{1}{\snr
    g(\bh,u)+1}\;du\bigg]^{-2},
\end{align}
where (\ref{eq:lemma:linear2}) was obtained by removing a finite-number of points from the integral in (\ref{eq:lemma:linear1}).
\begin{theorem}
\label{th:monotone}
\emph{Monotone Convergence}~\cite[Theorem. 4.6]{bartle:B1}:
if a function $F(u,v)$ defined on $U\times[a,b]\rightarrow \mathbb{R}$, is positive and monotonically increasing in $v$, and there exists an integrable function $\hat F(u)$, such that $\lim_{v\rightarrow\infty}F(u,v)=\hat F(u)$, then
\begin{equation}
    \label{eq:lemma:linear3} \underset{v\rightarrow\infty}\lim\int_U
    F(u,v)\;du=\int_U\underset{v\rightarrow\infty}\lim F(u,v
    )\;du=\int_U \hat{F}(u)\;du.
\end{equation}
\end{theorem}
For further simplifying (\ref{eq:lemma:linear2}), we define $F_1(u,\snr)$ and $F_2(u,\snr)$ over $\Big\{u\med u\in[-\pi,\pi], g(\bh,u)\neq  0\Big\}\times[1,+\infty]$ as
\begin{eqnarray*}
    F_1(u,\snr)&\dff&\frac{1}{g(\bh,u)}-\frac{1}{\snr^2
    g(\bh,u)}+\frac{g(\bh,u)}{(\snr g(\bh,u)+1)^2},\\
    \mbox{and}\;\;\;F_2(u,\snr)&\dff&\frac{1}{g(\bh,u)}-\frac{1}{\snr
    g(\bh,u)}+\frac{1}{\snr g(\bh,u)+1}.
\end{eqnarray*}
It can be readily verified that $F_i(u,\snr) > 0$ and $F_i(u,\snr)$ is increasing in $\snr$. Moreover, there exist $\hat F(u)$ such that
\begin{equation*}
    \hat F(u)=\lim_{\snr\rightarrow\infty}F_1(u,\snr)=\lim_{\snr\rightarrow\infty}F_2(u,\snr)=\frac{1}{g(\bh,u)}.
\end{equation*}
Therefore, by exploiting the result of  Theorem~\ref{th:monotone} we find
\begin{eqnarray*}
    &&\lim_{\snr\rightarrow\infty}\int\bigg[\frac{1}{g(\bh,u)}-\frac{1}{\snr^2
    g(\bh,u)}+\frac{g(\bh,u)}{(\snr g(\bh,u)+1)^2}\bigg]du=\int\frac{1}{g(\bh,u)}\;du,\\
    \mbox{and}&&\lim_{\snr\rightarrow\infty}\int\bigg[\frac{1}{g(\bh,u)}-\frac{1}{\snr
    g(\bh,u)}+\frac{1}{\snr
    g(\bh,u)+1}\bigg]du=\int\frac{1}{g(\bh,u)}\;du,
\end{eqnarray*}
or equivalently,
\begin{eqnarray}
    \label{eq:lemma:linear4}
    &&\lim_{\snr\rightarrow\infty}\frac{1}{2\pi}\int\frac{g(\bh,u)\;du}{(\snr
    g(\bh,u)+1)^2}=\lim_{\snr\rightarrow\infty}\frac{1}{2\pi}\int\frac{\;du}{\snr^2
    g(\bh,u)},\\
    \label{eq:lemma:linear5}
    \mbox{and}&&\lim_{\snr\rightarrow\infty}\frac{1}{2\pi}\int\frac{\;du}{\snr
    g(\bh,u)+1}=\lim_{\snr\rightarrow\infty}\frac{1}{2\pi}\int\frac{\;du}{\snr
    g(\bh,u)}.
\end{eqnarray}
By using the equalities in (\ref{eq:lemma:linear4})-(\ref{eq:lemma:linear5}) and proper replacement in (\ref{eq:lemma:linear2}) we get
\begin{align}
    \nonumber\lim_{\snr\to\infty} & \frac{\partial\;\mmsesnr}{\partial\;\snr}\\
    &=\lim_{\snr\to\infty}\bigg[\frac{1}{2\pi}\int_{g(\bh,u)\neq
    0}\frac{g(\bh,u)}{\left(\snr
    g(\bh,u)+1\right)^2}\;du\bigg]\cdot \bigg[\frac{1}{2\pi}\int_{g(\bh,u)\neq
    0}\frac{1}{\snr g(\bh,u)+1}\;du\bigg]^{-2}\\
    \nonumber &=\lim_{\snr\to\infty}\bigg[\frac{1}{2\pi}\int_{g(\bh,u)\neq
    0}\frac{1}{\snr^2
    g(\bh,u)}\;du\bigg]\cdot \bigg[\frac{1}{2\pi}\int_{g(\bh,u)\neq
    0}\frac{1}{\snr g(\bh,u)}\;du\bigg]^{-2}\\
    \nonumber &= \bigg[\frac{1}{2\pi}\int_{g(\bh,u)\neq
    0}\frac{1}{ g(\bh,u)}\;du\bigg]^{-1}=s(\bh),
\end{align}
where $s(\bh)$ is independent of $\snr$ and thus the proof is completed.

\renewcommand\url{\begingroup\urlstyle{rm}\Url}

\bibliographystyle{IEEEtran}
\bibliography{IEEEabrv,IIR_LE}

\begin{thebibliography}{10}
\providecommand{\url}[1]{#1}
\csname url@samestyle\endcsname
\providecommand{\newblock}{\relax}
\providecommand{\bibinfo}[2]{#2}
\providecommand{\BIBentrySTDinterwordspacing}{\spaceskip=0pt\relax}
\providecommand{\BIBentryALTinterwordstretchfactor}{4}
\providecommand{\BIBentryALTinterwordspacing}{\spaceskip=\fontdimen2\font plus
\BIBentryALTinterwordstretchfactor\fontdimen3\font minus
  \fontdimen4\font\relax}
\providecommand{\BIBforeignlanguage}[2]{{%
\expandafter\ifx\csname l@#1\endcsname\relax
\typeout{** WARNING: IEEEtran.bst: No hyphenation pattern has been}%
\typeout{** loaded for the language `#1'. Using the pattern for}%
\typeout{** the default language instead.}%
\else
\language=\csname l@#1\endcsname
\fi
#2}}
\providecommand{\BIBdecl}{\relax}
\BIBdecl

\bibitem{Proakis:book}
J.~Proakis, \emph{Digital Communications}, 3rd~ed.\hskip 1em plus 0.5em minus
  0.4em\relax McGraw Hill, 1995.

\bibitem{forney:ML}
J.~G.~D.~Forney, ``Maximum-likelihood sequence estimation of digital sequences
  in the presence of intersymbol interference,'' \emph{{IEEE} Trans. Inform.
  Theory}, vol.~18, no.~3, pp. 363--378, May 1972.

\bibitem{qureshi:adaptive}
S.~Qureshi, ``Adaptive equalization,'' \emph{Proceedings of the IEEE}, vol.~9,
  no.~73, pp. 1349-- 1387, September 1985.

\bibitem{vitetta}
D.~P. Taylor, G.~M. Vitetta, B.~D. Hart, and A.~Mämmelä, ``Wireless channel
  equalization,'' \emph{Eur. Trans. Telecommun}, vol.~9, no.~2, pp. 117--143,
  1998.

\bibitem{ali:ISIT07_1}
A.~Tajer, A.~Nosratinia, and N.~Al-Dhahir, ``{MMSE} infinite length
  symbol-by-symbol linear equalization achieves full diversity,'' in
  \emph{Proc. IEEE International Symposium on Information Theory (ISIT'07)},
  Nice, France, June 2007, pp. 1706--1710.

\bibitem{cioffi}
J.~M. Cioffi, \emph{Digital Communication: Signal Processing},
  \url{http://www.stanford.edu/class/ee379a/course_reader/chap3.pdf}.

\bibitem{CDEF}
J.~M. Cioffi, G.~Dudevoir, M.~Eyuboglu, and J.~G.D.~Forney, ``Minimum
  mean-square-error decision feedback equalization and coding - parts {I} and
  {II},'' \emph{{IEEE} Trans. Commun.}, vol.~43, no.~10, pp. 2582--2604,
  October 1998.

\bibitem{zheng:IT03}
L.~Zheng and D.~Tse, ``Diversity and multiplexing: A fundamental tradeoff in
  multiple antenna channels",,'' \emph{{IEEE} Trans. Inform. Theory}, vol.~49,
  no.~5, pp. 1073--1096, May 2003.

\bibitem{powell:book}
M.~J.~D. Powell, \emph{Approximation Theory and Methods}.\hskip 1em plus 0.5em
  minus 0.4em\relax Cambridge University Press, 1981.

\bibitem{mcgehse}
O.~C. Mc{G}ehee, L.~Pigno, and B.~Smith, ``Hardy's inequality and the $\ell_1$
  norm of exponential sums,'' \emph{Anals of Matehmatics}, vol. 113, pp.
  613--618, 1981.

\bibitem{bartle:B1}
R.~G. Bartle, \emph{The Elements of Integrations}.\hskip 1em plus 0.5em minus
  0.4em\relax John Wiley \& Sons Inc, 1966.

\bibitem{azarian:IT05}
K.~Azarian, H.~El-Gamal, and P.~Schniter, ``On the achievable
  diversity-multiplexing tradeoff in half-duplex cooperative channels,''
  \emph{{IEEE} Trans. Inform. Theory}, vol.~51, no.~12, pp. 4152--4172, Dec.
  2005.

\end{thebibliography}

\begin{figure}
  \centering
  \includegraphics[width=4.5 in]{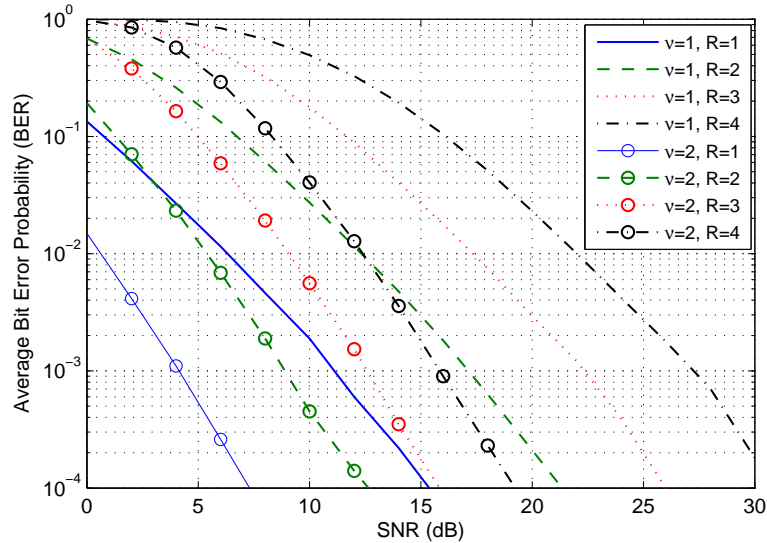}\\
  \caption{Average detection error probability in  two-tap and three-tap ISI channels with MMSE linear  equalization.}\label{fig:1}
\end{figure}
\begin{figure}
  \centering
  \includegraphics[width=4.5 in]{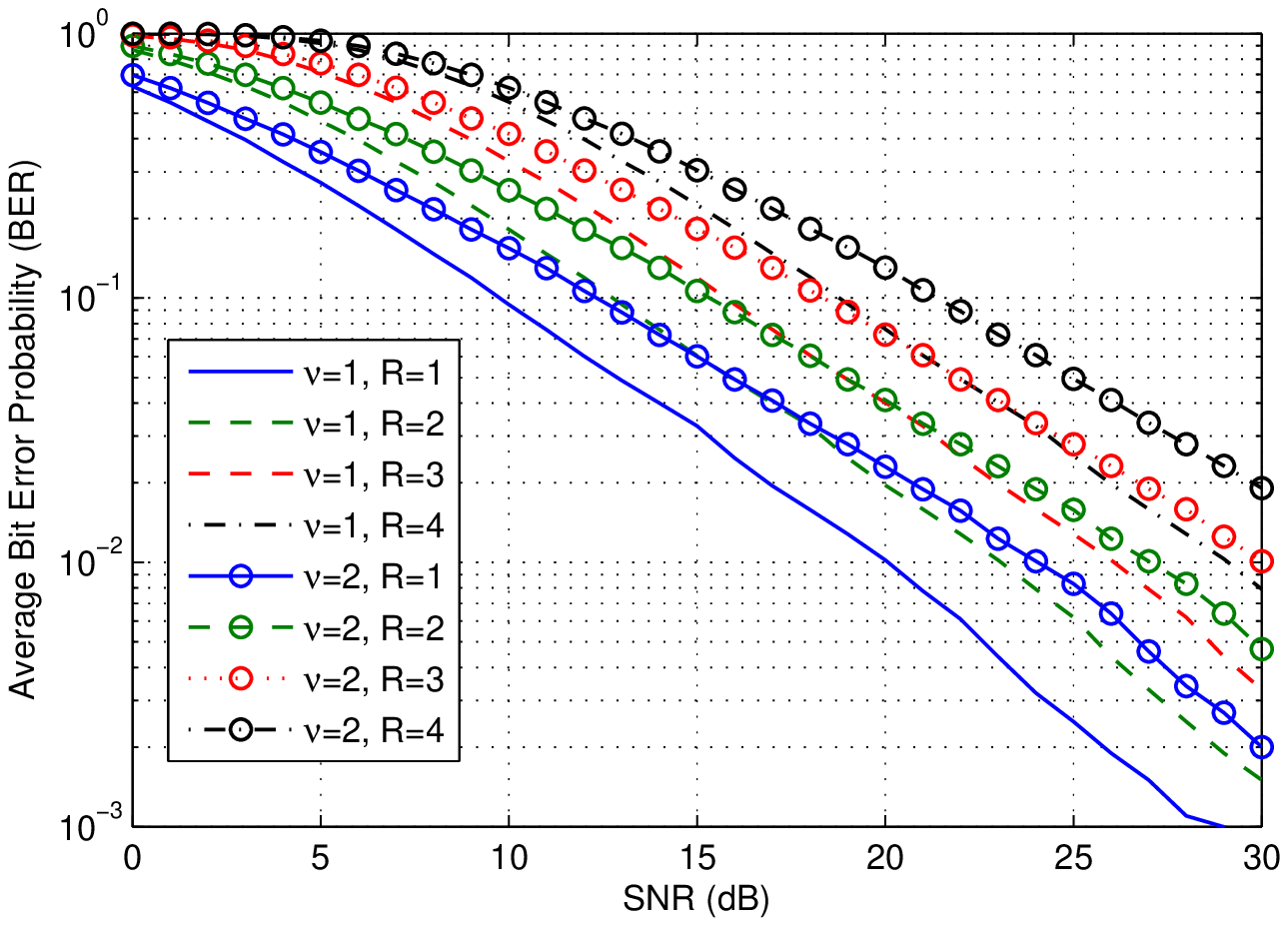}\\
  \caption{Average detection error probability in  two-tap and three-tap ISI channels with ZF linear equalization}\label{fig:2}
\end{figure}

\end{document}